\documentclass[11pt,a4paper]{article}

\usepackage[T1]{fontenc}
\usepackage[linesnumbered,ruled,vlined,noresetcount]{algorithm2e}
\DontPrintSemicolon
\SetArgSty{textnormal}
\SetKwProg{Fn}{function}{:}{end}
\SetKwProg{pFor}{for each}{ do in parallel}{end}
\SetKwProg{When}{when}{ do}{end}
\SetKwInput{Variables}{Variables}
\SetKwInput{Notation}{Notation}
\SetKwInput{Note}{Note}

\def\protocol#1{{\normalfont\textsc{#1}}\xspace}
\newcommand{\approximate}{\protocol{Approximate}}

\usepackage{graphicx}

\usepackage{multirow}

\usepackage{color}

\usepackage{authblk}
\usepackage{geometry}
\usepackage{amsmath,amssymb,amsthm}
\usepackage{graphicx}
\usepackage{hyperref}

\newtheorem{lemma}{Lemma}
\newtheorem{theorem}{Theorem}
\newtheorem{remark}{Remark}
\newtheorem{proposition}{Proposition}
\theoremstyle{definition}

 \newtheorem{definition}{Definition}

\newcommand{\cold}{\protocol{CollisionDetection}}
\newcommand{\coldbfull}{\protocol{CollisionDetection\allowbreak{}WithBounds}}
\newcommand{\coldb}{\protocol{CDWB}}
\newcommand{\phaseclock}{\protocol{PhaseClock}}

\newcommand{\epidemic}{\protocol{Epidemic}}

\newcommand{\collision}{\mathtt{collision}}
\newcommand{\rank}{\mathtt{rank}}

\newcommand{\epoch}{\mathtt{epoch}}
\newcommand{\timer}{\mathtt{timer}}
\newcommand{\grid}{\mathtt{gid}}

\newcommand{\children}{\mathtt{infectivity}}
\newcommand{\countfin}{\mathtt{countFin}}
\newcommand{\level}{\mathtt{level}}
\newcommand{\lognum}{\mathtt{logNum}}
\newcommand{\var}{\mathtt{var}}
\newcommand{\leader}{\mathtt{leader}}

\newcommand{\aaa}{a}
\newcommand{\bbb}{b}

 \newcommand{\nlower}{n_{L}}
 \newcommand{\nupper}{n_{U}}

\newcommand{\outputs}{f_{\mathrm{out}}} 
\newcommand{\sinit}{s_{\mathrm{init}}}
\newcommand{\tdone}{t_{\mathrm{done}}}
\newcommand{\rs}{\mathbf{\Gamma}} 

\newcommand{\cinit}{\mathcal{C}_{\mathrm{init}}}

\newcommand{\ie}{i.e.}

\newcommand{\etal}{et~al.}

\newcommand{\poly}{\mathit{poly}}

\newcommand{\childmax}{\left \lfloor  \log \nlower - \log \ell \right \rfloor-2}

\title{
Sublinear-time Collision Detection with a Polynomial Number of States in Population Protocols
} 
\date{}

\author{Takumi Araya}
\author{Yuichi Sudo}

\affil{Hosei University, Tokyo, Japan}

\begin{document}

\maketitle


\begin{abstract}
This paper addresses the collision detection problem in population protocols. The network consists of state machines called agents. At each time step, exactly one pair of agents is chosen uniformly at random to have an interaction, changing the states of the two agents. The collision detection problem involves each agent starting with an input integer between $1$ and 
$n$, where $n$ is the number of agents, and requires those agents to determine whether there are any duplicate input values among all agents. Specifically, the goal is for all agents to output false if all input values are distinct, and true otherwise.

In this paper, we present an algorithm that requires a polynomial number of states per agent and solves the collision detection problem with probability one in sub-linear parallel time, both with high probability and in expectation. To the best of our knowledge, this algorithm is the first to solve the collision detection problem using a polynomial number of states within sublinear parallel time, affirmatively answering the question raised by Burman, Chen, Chen, Doty, Nowak, Severson, and Xu [PODC 2021] for the first time.
\end{abstract}

\section{Introduction}
\label{sec:introduction}
In this paper, we explore the \emph{population protocol} model, introduced in 2004 and studied extensively since then \cite{AAD+06,AAE08,AG15,AAG18,GS18,GSU18,SOK+20,SOI+20,BGK20,SEI+21}. The model consists of a network, or \emph{population}, of $n$ state-machines, referred to as \emph{agents}. At each time step, a pair of agents is selected uniformly at random to engage in an \emph{interaction} (\ie, pairwise communication), during which they update their states.
Agents are anonymous, \ie, they lack unique identifiers.
In the population protocol model, time complexity is commonly measured in \emph{parallel time}, defined as the number of time steps divided by the population size $n$ (the number of agents). In practice, interactions often occur simultaneously across the population; parallel time approximately captures this behavior. Throughout the remainder of this section (\ie, Section \ref{sec:introduction}), we assume parallel time when discussing time complexity.

In this paper, we study the collision detection problem in the population protocol model.
Before presenting our main results, we briefly explain prior studies on problems closely related to collision detection, namely self-stabilizing leader election and self-stabilizing ranking.

Leader election has been extensively studied in the population protocol model. Leader election can be solved by a simple two-state protocol \cite{AAD+06}, where initially, all agents are leaders. The protocol employs only one transition rule: when two leaders meet, one of them becomes a follower (\ie, a non-leader). This simple protocol elects a unique leader in $O(n)$ parallel time.
Furthermore, this protocol is time-optimal: Doty and Soloveichik \cite{DS18} demonstrated that any constant-space protocol requires linear time to elect a unique leader. Later, in 2015, Alistarh and Gelashvili \cite{AG15} developed a leader election protocol that converges in $O(\log^3 n)$ parallel time and uses $O(\log^3 n)$ states per agent.
Subsequently, numerous papers have focused on fast leader election, including \cite{AAG18,GS18,GSU18,SOI+20,BGK20}. G{\k{a}}sieniec, Staehowiak, and Uznanski \cite{GSU18} developed an algorithm that converges in $O(\log n \log \log n)$ time and uses a surprisingly small number of states: only $O(\log \log n)$ per agent. This is considered space-optimal because it is established that every leader election protocol requiring $O(n/\mathrm{polylog}(n))$ time also requires $\Omega(\log \log n)$ states \cite{AAE+17}. 
Sudo \etal~\cite{SOI+20} presented a simple protocol that elects a unique leader within $O(\log n)$ time and utilizes $O(\log n)$ states per agent. This is time-optimal, as any leader election protocol requires $\Omega(\log n)$ time, even if it uses an arbitrarily large number of states and the agents know the exact size of the population \cite{SM20}. (This lower bound may appear obvious, yet it does not directly follow from a simple coupon collector argument because we can specify an initial configuration where all agents are followers.)
Finally, in 2020, Berenbrink \etal~\cite{BGK20} provided a time and space-optimal protocol, \ie, an $O(\log n)$-time and $O(\log \log n)$-states leader election protocol.

Self-stabilizing leader election (SS-LE) has garnered significant attention within this model. This variant of leader election stipulates that (i) starting from any configuration, the population must reach a safe configuration where exactly one leader exists; and (ii) once a safe configuration is reached, the unique leader must be maintained indefinitely. These conditions ensure tolerance against finitely many transient faults, which is critical since many protocols (both self-stabilizing and non-self-stabilizing) assume the presence of a unique leader. Consequently, SS-LE is essential for enhancing the fault tolerance of the population protocol model itself.
However, no protocol can solve SS-LE unless each agent in the population knows the \emph{exact size} $n$ of the population \cite{FJ06,CIW12}\footnotemark{}. 
\footnotetext{ 
Strictly speaking, the cited works proves a slightly weaker impossibility. Nevertheless, this impossibility can be proved using a similar technique: a simple partitioning argument. See \cite{SOK+20} for details (page 618, footnote).
}

Numerous studies have focused on overcoming this impossibility by employing various strategies, including assuming oracles \cite{BBB13,FJ06,CP07}, assuming that agents precisely know the population size $n$ \cite{CIW12,BCC+21}, restricting the topology \cite{AAF+08,CC19,CC20,YSM21,YSO+23}, or slightly relaxing the requirements of SS-LE~\cite{Izu15,SOK+14,SMD+16,SNY+12,SOK+20det,SOK+18,SOK+20,SSN+21,SEI+21}. Among these, all algorithms that adopt the $n$-knowledge approach \cite{CIW12,BCC+21} elect the unique leader by solving a more general problem, called \emph{self-stabilizing ranking}. This problem stipulates that: (i) each agent $v$ maintains an output variable $v.\rank \in \{1,2,\dots,n\}$,
(ii) starting from any configuration, the population must reach a safe configuration where no two agents share the same $\rank$ value; and (iii) once a safe configuration is reached, no agent updates its $\rank$. SS-LE can be straightforwardly reduced to self-stabilizing ranking because, once ranking is achieved, there is exactly one agent with $\rank=1$. This agent can thus be regarded as the unique leader.
Cai, Izumi, and Wada~\cite{CIW12} present an algorithm that solves the self-stabilizing ranking problem within $O(n^2)$ parallel time, using $n$ states per agent. In contrast, Burman, Chen, Chen, Doty, Nowak, Severson, and Xu~\cite{BCC+21} introduce significantly faster algorithms that, however, require more states. Specifically, they offer an $O(n)$ parallel time algorithm with $O(n)$ states per agent and an $O(\log n)$ parallel time algorithm with a super-exponential number of states. This leads to a natural question: Can self-stabilizing ranking be solved in sublinear parallel time using only a polynomial number of states per agent?
To solve self-stabilizing ranking, agents must detect whether any distinct agents share the same rank. Consequently, Burman \etal\ raise an open question: Is there a sublinear parallel time algorithm with $\mathit{poly}(n)$ states that can solve the following problem, which we refer to as the \emph{collision detection problem} in this paper (See Section \ref{sec:prodef} for the formal definition
of this problem):
\begin{quote}
Each agent $v$ is given an input $\rank \in \{1,2,\dots,n\}$. The goal for each agent is to decide whether at least one pair of agents have the same input value in $\rank$.
\end{quote}

\subsection*{Our Contribution}
In this paper, we affirmatively answer the open question raised by Burman \etal~\cite{BCC+21}. Specifically, we introduce a collision detection algorithm that stabilizes within $O(\sqrt{n \log n})$ parallel time in expectation and $O(n^{1/2} \cdot \log^{3/2} n)$ parallel time with high probability. This algorithm uses $\tilde{O}(n)$ states per agent with high probability, excluding the input variable $\rank$, which requires $O(n)$ states.
The proposed algorithm is always correct; that is, it eventually reaches a stable configuration where all agents output the correct answer with probability 1.


It remains an open question whether self-stabilizing ranking, or weaker variants such as loosely-stabilizing ranking, can be solved within sublinear parallel time using a polynomial number of states per agent. This question persists primarily because our collision detection protocol is not self-stabilizing.

\subsection*{Organization of This Paper}
Section 2 introduces the preliminaries, including key terminologies and the definition of the model. Section 3 describes the basic submodules that are instrumental in designing the proposed algorithm $\cold$. Section 4 presents $\cold$, proves its correctness, and bounds its time and space complexities.

In the remainder of this paper, we will not use parallel time; instead, we will discuss stabilization time in terms of the number of time steps (or interactions).

\section{Preliminaries}
\label{sec:preliminaries} 


Throughout this paper,
we denote the set of \emph{non-negative} integers by $\mathbb{N}$
and represent the set
$\{z \in \mathbb{N} \mid x \leq z \leq y\}$ 
by $[x,y]$.
When the base of a logarithm is omitted, it is assumed to be 2.
For any positive integer $i$, we denote the $i$-th projection map by $\pi_i$,
\ie, for any element $\mathbf{x} = (x_1, x_2, \dots, x_k) \in X_1 \times X_2 \times \dots \times X_k$,
we define $\pi_i(\mathbf{x}) = x_i$.
The $\tilde{O}$-notation hides a poly-logarithmic factor; that is, for any function $f(x)$,
$\tilde{O}(f(x))$ denotes $O(f(x) \cdot \log^c x)$ for some constant $c$.

\subsection{Model}
\label{sec:model}

A \emph{population} is a network consisting of {\em agents}.
We denote the set of all the agents by $V$ and let $n = |V|$.
We assume that a population is a complete graph,
thus every pair of agents $(u,v)$ can interact,
where $u$ serves as the \emph{initiator}
and $v$ serves as the \emph{responder} of the interaction.

A \emph{protocol} $P(Q,\sinit,X,Y,T,\outputs)$ consists of 
a set $Q$ of states,
the initial state $\sinit$, 
a set $X$ of input symbols, 
a set $Y$ of output symbols, 
a transition function
$T:  (Q\times X)^2 \to Q^2$,
and an output function $\outputs : Q \times X \to Y$.
Initially, all agents are in state $\sinit \in Q$.
When two agents interact,
$T$ determines their next states
according to their current states and inputs.
The \emph{output} of an agent is determined by $\outputs$:
the output of an agent in state $q \in Q$ and input $x \in X$ 
is $\outputs(q,x)$.
A protocol may be given design parameters, such as a lower bound $\nlower$ and an upper bound $\nupper$ on the population size $n$. 
In this case, each component of the protocol, namely $Q$, $\sinit$, $X$, $Y$, $T$, and $\outputs$, may depend on these parameters.

A \emph{configuration} is a mapping $C : V \to Q\times X$ that specifies the states and the inputs
of all agents in the population. 
We say that a configuration $C$ \emph{changes} to another configuration $C'$ via an interaction $e = (u,v)$, 
denoted by $C \stackrel{P,e}{\to} C'$, 
if $(\pi_1(C'(u)), \pi_1(C'(v))) = T(C(u), C(v))$, $\pi_2(C'(u))=\pi_2(C(u))$, $\pi_2(C'(v))=\pi_2(C(v))$,
and $C'(w) = C(w)$ for all $w \in V \setminus \{u, v\}$.\footnote{
In this paper, we assume that the inputs provided to the agents do not change.
}
We say that a configuration $C'$ is reachable from a configuration $C$
if there is a sequence of configurations $C = C_0, C_1, \dots, C_k$ such that 
$C_i$ changes to $C_{i+1}$ via some interaction for each $i \in [0, k-1]$.
A configuration $C$ is \emph{stable} if, for every configuration $C'$ reachable from $C$,
the output of each agent remains the same in both $C$ and $C'$,
\ie, $\outputs(C(v)) = \outputs(C'(v))$.


We assume the \emph{uniformly random scheduler} $\rs$, which
selects two agents to interact at each step uniformly at random from all ordered pairs of agents.
Specifically, $\rs=\Gamma_0, \Gamma_1,\dots$
where each $\Gamma_t$ is a random variable such 
that $\Pr(\Gamma_t = (u,v)) =\frac{1}{n(n-1)}$ 
for any $t \ge 0$ and any distinct $u,v \in 
V$. 
The \emph{execution} of protocol $P$ starting from a configuration $C_0$
under the uniformly random scheduler $\rs$
is defined as the sequence of configurations $\Xi_{P}(C_0,\rs) = C_0,C_1,\dots$ such that
$C_t \stackrel{P,\Gamma_t}{\to} C_{t+1}$ holds for all $t \ge 0$. 
Note that each $C_i$ is also a random variable.

In this paper, we use the notation $v.\var$ to denote the value of a variable $\var$ maintained by an agent $v$. 
By abuse of notation,
we use $C(v).\var$ to represent the value of $\var$ in a configuration $C$.

\subsection{With High Probability}
\label{sec:whp}
In this paper, we frequently use the term ``with high probability,'' abbreviated as w.h.p. 
Various definitions for this term exist, such as with probability $1-n^{-1}$ or with probability $1-n^{\Omega(1)}$. 
Following \cite{BEF21}, we adopt the below definition.

\begin{definition}[with high probability]
\label{def:whp}
A property of a protocol $P$ holds \emph{with high probability}
if, for any constant $\eta = \Theta(1)$, we can ensure
the property is satisfied with a probability of $1-O(n^{-\eta})$
by adjusting the constant parameters of $P$, if necessary.
\end{definition}


\noindent 
For example, if a protocol $P_{\tau}$ with a design parameter $\tau=\Theta(1)$
elects a leader in $\tau \cdot n \log n = O(n \log n)$ time steps with probability $1-n^{-\tau}$,
we say that $P_{\tau}$ elects a leader in $O(n\log n)$ time steps w.h.p.
In this case, we can also explicitly state the constant factor of the running time and say that
$P_{\tau}$ elects a leader in $\tau \cdot n \log n$ time steps w.h.p.
Note that the requirement for the success probability (i.e., $1-O(n^{-\eta})$) is defined asymptotically;
therefore, we only need to consider sufficiently large population sizes $n$.


\subsection{Collision Detection Problem}
\label{sec:prodef}


The collision detection problem requires a protocol $P = (Q, \sinit, X, Y, T, \outputs)$ to accept the input symbol $[1,n]$, \ie, $[1,n] \subseteq X$, where $n = |V|$ is the population size. Each agent $v$ is assigned an input symbol in $[1,n]$, denoted by $v.\rank$.
A configuration $C : V \to Q \times X$ is termed \emph{initialized} if, for every $v \in V$, $v$ is in the initial state $\sinit$ and its input $v.\rank$ is within $[1,n]$ in $C$.
More formally, $C$ is initialized if $\pi_1(C(v)) = \sinit$ and $\pi_2(C(v)) \in [1,n]$ for all $v \in V$.
We define $\cinit(P)$ as the set of these initialized configurations.
The goal of this problem is to detect whether or not a \emph{collision} exists, \ie, $u.\rank = v.\rank$ for some distinct agents $u, v \in V$. Each agent must output either $1$ or $0$, representing ``there is at least one collision among the ranks of the agents'' and ``there is no such collision,'' respectively. For simplicity, we assume every protocol $P$ includes a variable $\collision \in \{0,1\}$, and an agent $v$ outputs $1$ if and only if $v.\collision = 1$.
That is, $\outputs(C(v)) = v.\collision$ for any configuration $C$ and any agent $v \in V$.

\begin{definition}[Collision Detection Problem]
Let $P=(Q,\sinit, X,Y,T,\outputs)$ be a protocol, where $[1,n] \subseteq X$. A protocol $P$ \emph{solves} the collision detection problem if for any initialized configuration $C_0 \in \cinit(P)$, both of the following conditions are satisfied:
\begin{itemize}
\item
For any stable configuration $C$ reachable from $C_0$, 
$C(v).\collision = 1$ for all $v \in V$ if $C_0(u).\rank = C_0(w).\rank$ for some distinct agents $u,w \in V$; otherwise, $C(v).\collision = 0$  holds for all $v \in V$.
\item The execution $\Xi_{P}(C_0,\rs)$ eventually reaches a stable configuration with probability 1.
\end{itemize}
\end{definition}

For any execution $\Xi = C_0, C_1, \dots$, we define the \emph{stabilization time} of $\Xi$
as the minimum $t$ such that $C_t$ is a stable configuration. 
We say that a protocol $P$ \emph{stabilizes} within $t$ time steps with high probability (respectively, in expectation)
if, for any initialized configuration $C_0 \in \cinit(P)$,
the stabilization time of $\Xi_{P}(C_0,\rs)$ is at most $t$ with high probability (respectively, in expectation).

\subsection{Uniform Protocols}
\label{sec:uniform}
A protocol is \emph{uniform} if it does not depend on the population size at all,
that is, if it does use any knowledge on the population size.
In the field of population protocols, originally, 
the number of states of a protocol $\mathcal{P}=(Q,\rho,Y,\pi,\delta)$ is simply defined as $|Q|$.
Thus, the number of states for uniform protocols must be either $\Theta(1)$ or infinite.
Some problems inherently require memory space of a non-constant size.\footnote{
For example, Doty and Soloveichik \cite{DS18} proved that $\omega(1)$ states are necessary to solve the leader election problem within $o(n^2)$ expected interactions under the uniformly random scheduler.
}
This means that uniform protocols devoted for those problems always require an infinite number of states by definition,
while each execution of these protocols may use only small number of states depending on population size. 
(A uniform protocol must be defined independently from the population size, but of course, its execution may depend on
the population size.)
Thus, we require another way to evaluate the number of states for uniform protocols as several works do
\cite{CMN+11,BKR19,DE19,BEF21}. This paper adopts the following simple definition of \emph{the number of states} of
a protocol $P=(Q,\sinit, X,Y,T,\outputs)$. 
\begin{itemize}
    \item An agent maintains a constant number of variables $x_1, x_2, \dots, x_s$, and the combination of their values constitutes the state of the agent. The first element $Q$ of protocol $P=(Q,\sinit, X,Y,T,\outputs)$ can be regarded as the set of all such states.
    The domain of each variable may not be bounded, thus $Q$ may be an infinite set.    
    \item For any agent $v \in V$, define $v$'s amount of information at time step $t$ as $f(v,t)=\sum_{i=1}^s \iota_i$, where each $\iota_i$ is the number of bits required to encode the value of variable $x_i$. The number of states of an execution $\Xi$ is defined as $\#(\Xi) = \max\{2^{f(v,t)} \mid v \in V, t=0,1,\dots\}$ in the execution.
    We say that the number of states of a protocol $P$
    is at most $z$ with high probability (respectively, in expectation) if 
    for any initialized configuration $C_0 \in \cinit(P)$,
    $\#(\Xi_{P}(C_0,\rs)) \le z$ holds with high probability (respectively, in expectation).
\end{itemize}

\section{Tools}
\label{sec:tools}
\begin{algorithm}[t]
\caption{$\epidemic(\var)$
at an interaction where
$\aaa$ and $\bbb$ are an initiator and a responder, respectively.
} 
\label{al:epidemic}
$\bbb.\var \gets \max(\aaa.\var,\bbb.\var)$
\end{algorithm}
\subsection{One-way Epidemic}
The one-way epidemic protocol was introduced and analyzed by Angluin, Aspnes, and Eisenstat \cite{AAE08} and has been widely used thereafter. 
The goal of this protocol is to propagate the maximum value of a given variable $\var$ to all agents.
The strategy is simple: when an initiator $a$ and a responder $b$ meet, and if $a.\var > b.\var$ holds,
$b.\var$ is updated to $a.\var$. (See Algorithm \ref{al:epidemic}.)
Angluin \etal~\cite{AAE08} prove the following useful lemma,
which guarantees that
the maximum value is propagated to the entire population
in $O(n \log n)$ steps w.h.p.

\begin{lemma}
\label{lem:epidemic}
Suppose that 
an execution of the protocol $\epidemic(\var)$
starts from a configuration $C$ where $M = \max_{v \in V} v.\var$.
Then, for any fixed $\eta > 0$, 
there exists a constant $d$ such that 
for sufficiently large population size $n$, 
with probability $1-n^{-\eta}$,
every agent $v \in V$ satisfies $v.\var \ge M$
within $d n \log n$ time steps.
\end{lemma}

\begin{algorithm}[t]
\caption{$\phaseclock(F)$
at an interaction where
$\aaa$ and $\bbb$ are an initiator and a responder, respectively.
} 
\label{al:phaseclock}
\Variables{$\timer \in [0,m-1], \epoch\in [0,F], \leader \in \{0,1\}$}
\Note{$m$ is a sufficiently large constant. (See Lemma \ref{lem:phase_clock}.)}
\Notation{
$\max_m(x,y)=
\begin{cases}
    x&  x \in [0,y-\lceil m/2 \rceil] \cup [y+1,y+\lfloor m/2 \rfloor]\\
    y& \textbf{otherwise}
\end{cases}$
}
\uIf{$\bbb.\leader = 1 \land \aaa.\timer = \bbb.\timer$}{
$\bbb.\timer \gets (\bbb.\timer + 1) \bmod m$\;
\If{$\bbb.\timer = 0$}{$\bbb.\epoch = \min(\bbb.\epoch+1,F)$\;}
}\ElseIf{
$\bbb.\leader=0$
}{
$\bbb.\timer \gets \max_m(\aaa.\timer,\bbb.\timer)$\;
}
Execute $\epidemic(\epoch)$\;

\end{algorithm}

\subsection{Phase Clock with a Leader}
Together with the epidemic protocol, Angluin \etal~\cite{AAE08}
introduce the phase clock protocol. 
This protocol requires the presence of a unique leader.
Specifically, it assumes that each agent maintains a variable $\leader \in \{0,1\}$,
and there exists an agent $a_L \in V$ such that
$a_L.\leader = 1$ and $b.\leader = 0$ for all $b \in V \setminus \{a_L\}$ consistently.
The goal of this protocol 
is allowing the leader $a_L$ to determine whether $\Theta(n \log n)$ time steps has passed or not.
In this paper, we modify this protocol slightly so that 
all agents are synchronized with a variable $\epoch \in [0,F]$,
where $F$ is a given (possibly non-constant) integer. 

In addition to $\epoch$, this protocol maintains a variable $\timer \in [0, m-1]$, where $m$ is a constant parameter.
Initially, $a_L.\timer = 0$ and $v.\timer = m-1$ for all $v \in V \setminus \{a_L\}$,
while the $\epoch$ value of all agents are $0$. 
In this protocol, only the responder updates its variables at each interaction.
Lines 2--8 in Algorithm \ref{al:phaseclock} specify how the responder updates $\timer$ and $\epoch$.
The leader $a_L$ increments its $\timer$ by one modulo $m$ when it encounters an initiator
with the same timer value. 
The timer of a non-leader $b$ is overwritten when
it encounters an initiator with a $\timer$ value from $b.\timer+1$
to $b.\timer + \lfloor m/2 \rfloor$ modulo $m$,
\ie, a $\timer$ value in $[0, y-\lceil m/2 \rceil] \cup [y+1, y+\lfloor m/2 \rfloor]$,
where $y=b.\timer$.
The leader $a_L$ increments its $\epoch$ by one when it increments its $\timer$ from $m-1$ to $0$
unless $a_L.\epoch$ has already reached the given maximum value $F$.
The increased $\epoch$ value is propagated to all agents via the epidemic protocol.

The following lemma directly follows from Lemma \ref{al:epidemic}
and the analysis of the phase clock protocol presented by Angluin \etal~\cite{AAE08}.

\begin{lemma}
\label{lem:phase_clock}
Let $C_0$ be a configuration with exactly one leader $a_L$, where 
the variables of all agents are initialized as specified above.
Let $\Xi_m = \Xi_{\phaseclock(F)}(C_0,\rs)$,
where $\phaseclock$ is the protocol specified in Algorithm \ref{al:phaseclock}
with the parameter $m$.
For any fixed constants $c, d_1, \eta \ge 0$,
there exist constants $m$ and $d_2$ such that,
for all sufficiently large population sizes $n$, 
the execution $\Xi_m$ satisfies the following properties with probability $1-n^{-\eta}$:
\begin{itemize}
\item For any $i \in [0,\min(n^c, F)]$, there are at least $d_1 \log n$ time steps where $v.\epoch=i$ holds for all agents $v$, and
\item For any $i \in [0,\min(n^c, F)-1]$, each agent $v$ maintains $v.\epoch = i$ for at most $d_2n \log n$ steps.
\end{itemize}
\end{lemma}

By this lemma, as long as $F=\poly(n)$, 
we can ensure that with high probability, all agents simultaneously experience each of the epochs $0, 1, \dots, F$ during an arbitrarily large $\Omega(n \log n)$ time steps with a controllable constant factor $d_1$,
while each epoch $0,1,\dots,F-1$ finishes in $O(n \log n)$ time steps.



\section{Collision Detection}
In this section, we present a collision detection protocol $\cold$, which
stabilizes within $O(n^{3/2} \cdot \sqrt{\log n})$ time steps (respectively, $O(n^{3/2} \cdot \log^{3/2} n)$ time steps) in expectation
(respectively, with high probability) and uses $\tilde{O}(n)$ states per agent with high probability, excluding the input variable $\rank$, which requires $O(n)$ states.
This protocol is always correct,
\ie, it eventually reaches a stable configuration where all agents output the correct answer
with probability $1$.

In Section \ref{sec:without}, we present a collision detection protocol under the assumption that (i) there is a pre-elected leader, and (ii) all agents are aware of a common lower bound and a common upper bound on $n$, both of which are asymptotically tight. In Section \ref{sec:with}, we eliminate these assumptions.

\subsection{Collision Detection with a Leader and Rough knowledge of $n$}
\label{sec:without}
In this section, we present a protocol $\coldbfull(\nlower, \nupper)$,
abbreviated as $\coldb(\nlower, \nupper)$.
This protocol requires that there is always a unique leader $a_L \in V$
and that
the two arguments $\nlower$ and $\nupper$ are asymptotically tight lower and upper bounds on $n$.
Specifically, $c_1 n \le \nlower \leq n \leq \nupper \le c_2 n$ must hold
for some constants $c_1$ and $c_2$.
We assume that the agents are aware of these constants $c_1$ and $c_2$, allowing us to incorporate them into the design of this protocol.

\begin{algorithm}[t]
\caption{$\coldb(\nlower,\nupper)$ at an interaction where
$\aaa$ and $\bbb$ are an initiator and a responder, respectively.
} 
\label{al:coldb}
\If{$\aaa.\rank = \bbb.\rank$}{$\bbb.\collision \gets 1$\;}
Let $r = \lceil 3 \eta \log \nupper \rceil$, $\ell = \left \lceil \sqrt{\nlower \log \nlower} \right \rceil$, and $z = \lceil \nupper / \ell \rceil$\;
Execute $\phaseclock(r\cdot z +1)$\;
\If{$\bbb.\epoch$ has increased in $\phaseclock$}{
Let $k=(\bbb.\epoch-1)\bmod z$\;
\uIf{$\bbb.\rank \in [k\cdot \ell+1,(k+1)\cdot \ell]$}{
Choose $\chi$ uniformly at random from $\{0,1\}$\;
$\bbb.\grid \gets (\bbb.\rank - k\cdot \ell,\chi)$\;
$\bbb.\children \gets \childmax$\;
}\Else{
$\bbb.\grid \gets \bot$\;
$\bbb.\children \gets 0$\;
}
}
\If{
$\aaa.\epoch = \bbb.\epoch > 0$
}{
\If{$\aaa.\grid \neq \bot \land \aaa.\children > 0 \land \bbb.\grid = \bot$}{
$\bbb.\grid \gets \aaa.\grid$\;
$\aaa.\children \gets \bbb.\children \gets \aaa.\children-1$
}
\If{$\exists x,y,z: \aaa.\grid=(x,y) \land \bbb.\grid=(x,z) \land y\neq z$}{
$\bbb.\collision \gets 1$\;
}
}
Execute $\epidemic(\collision)$
\end{algorithm}

The pseudocode of this protocol is presented in Algorithm \ref{al:coldb}.
The algorithm consists of three parts.
The first part, referred to as \emph{the backup protocol}, detects collisions through direct interactions (Lines 9--10):
when two agents with the same rank meet, the responder $\bbb$ raises a collision flag,
\ie, $\bbb.\collision$ is set to $1$.
If only one pair of agents $u, v \in V$ shares a common rank,
it can be easily observed that this process requires $\binom{n}{2} = \Theta(n^2)$ time steps in expectation.
The second part detects collisions much more quickly, requiring $O(n^{3/2} \sqrt{\log n})$ time steps in expectation
and $O(n^{3/2} \log^{3/2} n)$ time steps with high probability, using a square root decomposition technique (Lines 11-27).
Both the first and second parts allow \emph{some} agent to raise a collision flag if a collision is detected.
The third part propagates the raised flag to the entire population via the epidemic protocol (Lines 28).

Although the backup protocol is much slower,
it is still necessary because the second part may fail to detect a collision with non-zero probability.
Moreover, the backup protocol assists the second part in quickly detecting collisions.
Specifically, we present the following lemma:
\begin{lemma}
\label{lem:not_so_many_colliding_pairs}
Let $C_0$ be a configuration where there are $x = \Omega(\sqrt{n\log n})$ pairs of colliding agents,
\ie, $|\{(u,v) \in V^2 \mid  u \neq v, u.\rank = v.\rank\}| = \Omega(\sqrt{n\log n})$ holds in $C_0$.
Then, in an execution of the backup protocol starting from $C_0$,
at least one agent raises a collision flag within $O(n^{3/2})$ time steps, both in expectation and with high probability.
\end{lemma}

\begin{proof}
At each time step, one of the $x$ colliding pairs is selected to interact
with a probability of $x / \binom{n}{2} = \Omega(\log n / n^{3/2})$, resulting in the responder of the pair raising a collision flag. 
Therefore, a collision flag is raised within $O(n^{3/2} / \log n)$ steps in expectation,
and within $O(n^{3/2})$ steps with high probability.
 \end{proof}
\noindent
This lemma implies that, regardless of the performance of the second part, 
the proposed protocol solves the collision detection problem within $O(n^{3/2})$ time steps, both in expectation and with high probability, if the number of colliding pairs is $\Omega(\sqrt{n \log n})$.
Therefore, for the remainder of this section, we will assume that the number of colliding pairs is $o(\sqrt{n \log n})$.

\begin{figure}[t]
 \centering 
 \includegraphics[width=0.85\textwidth]{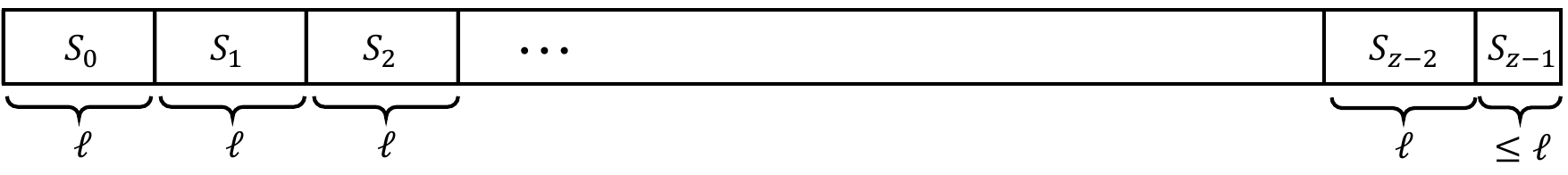}
 \caption{Segments of the second part of $\coldb(\nlower,\nupper)$} 
 \label{fig:segments}
\end{figure}

\begin{table}
\centering
\caption{Notations used in the second part of $\coldb(\nlower,\nupper)$}
\label{table:notations}
\begin{tabular}{c c c}
\hline
Notation & Definition & Explanation \\
\hline 
$\ell$ & $\lceil \sqrt{\nlower \log \nlower} \rceil$ & Length of a segment \\
$z$ & $\lceil \nupper/\ell \rceil$ & Number of segments \\
$S_k$ & $[k\cdot \ell + 1, (k+1)\cdot \ell] \cap [1, \nupper]$ & $k$-th segment
($0 \le k < z$)\\
$r$ & $\lceil 3 \eta \cdot \log \nupper \rceil$ & Number of epochs dedicated to each segment \\
\hline
\end{tabular}
\end{table}

We will now explain how the agents behave in the second part, which is the main focus of the proposed protocol.
Utilizing the unique leader assumption, the agents selected for each interaction
execute the phase clock protocol presented in Section \ref{sec:tools} (Line 12).
Let $\ell=\left \lceil \sqrt{\nlower \log \nlower} \right \rceil$
and $z = \lceil \nupper/\ell \rceil$.
In addition, let $r = \lceil 3 \eta \cdot \log \nupper \rceil$, where $\eta$ is a hidden design parameter ensuring that the targeted time complexity, \ie, $O(n^{3/2} \cdot \log^{3/2} n)$ time steps, are achieved with probability $1-O(n^{-\eta})$.
The argument $F$ of the phase clock protocol is set to
$r \cdot z + 1$. By Lemma \ref{lem:phase_clock}, it is guaranteed with high probability that
for each $i = 1, 2, \dots, r\cdot z$, 
all agents are simultaneously in epoch $i$
for at least $d_1 n \log n$ steps with any sufficiently large constant $d_1$,
and all agents enter the final epoch $r \cdot \ell + 1$ within
$\Theta(r z n \log n)=\Theta(n^{3/2} \log^{3/2} n)$ steps.
We divide the set of integers $[1,\nupper]$ into $z$ segments with at most size $\ell$,
\ie, for any $k=0,1,2,\dots,z-1$,
the $k$-th segment is $S_k = [k\cdot \ell+1, (k+1)\cdot \ell] \cap [1,\nupper]$
(Figure \ref{fig:segments}).
We assign $r$ epochs $k-1, k-1 + z, k-1 + 2z, \dots, k-1 + (r-1)z$ to
detect collisions of ranks in $S_k$.
Conversely, each epoch $i \in [1, r\cdot z]$ is dedicated to collision detection in $S_{(i-1) \bmod z}$.
The notations introduced here are summarized in Table \ref{table:notations}.

The goal of epoch $i \in [1,r\cdot z]$ is to detect whether or not there exists $x \in S_{(i-1) \bmod z}$ such that 
$|\{v \in V \mid v.\rank = x\}| \ge 2$,
in $O(n\log n)$ time steps with probability at least $1/3$ for a sufficiently large $n$, while 
ensuring that no agent mistakenly detect a collision if no such $x$ exists. 
By the achievement of this goal, all $r$ epochs devoted to each segment results in 
the correct decision with probability $1-(1/3)^r \ge 1-n^{-\eta}$.
To achieve this goal, each agent $v \in V$  maintains a variable $v.\grid \in ([1,\ell]\times \{0,1\}) \cup \{\bot\}$ and $v.\children \in [0,\childmax]$.
Suppose that an agent $v$ enters epoch $i$ now.
If $v.\rank \in S_{k}$, where $k=(i-1) \bmod z$, it generates a group identifier
$(v.\rank-k\cdot\ell, \chi)$, where $\chi$ is a random number (or \emph{nonce})
chosen uniformly at random in $[0,\lceil {\nupper}^{\eta/r}\rceil-1]$,
and stores it in $v.\grid$ (Lines 15--18). In addition, $v.\children$ is reset to the maximum value $\childmax$.
Otherwise, $v.\grid$ and $v.\children$ are set to $\bot$ and $0$, respectively
(Lines 19--21).
We say that an agent with $\grid=\bot$ a \emph{null} agent.
A non-null agent with $\children = x$ has ability to make $2^x$ copies of non-null agent:
whenever a non-null agent $a$ with $a.\children > 0$ and a null agent $b$ 
have an interaction as the initiator and the responder, respectively, 
$a.\grid$ is copied to $b.\grid$ and both $a.\children$ and $b.\children$
are set to $a.\children - 1$ (Lines 23--25).
If two agents have an interaction at which
their group identifier share a common rank,
but have different nonce, they notice the existence of a collision.
Then, the responder of the interaction raises a collision flag (Line 27).


The proliferation of group identifiers resembles the epidemic protocol \cite{AAE08}. However, unlike the epidemic protocol, this process limits the maximum number of non-null agents. This constraint
enhances the speed of proliferation, ensuring that proliferation completes in every epoch with high probability. This is formalized in the following lemma:

\begin{lemma}
\label{lem:proliferation}
Let $P = \coldb(\nlower, \nupper)$,
$i \in [1, r \cdot z]$, and $k = (i-1) \bmod z$.
Let $C_0 \in \cinit(P)$ be any initialized configuration
where there are at most $2 \cdot \ell$ agents whose rank is in $S_k$,
and let $\Xi_{P}(C_0, \rs) = C_0, C_1, \dots$.
Let $t_b$ be the minimum integer such that all agents are in epoch $i$ in the $t_b$-th configuration $C_{t_b}$. Then, with high probability,
there exists a non-negative integer $\tdone = t_b + O(n \log n)$ such that
for all $v \in V$, $v.\epoch = i$ and $v.\children = 0$ in $C_{\tdone}$.
\end{lemma}

\begin{proof}
Let $t_e$ be the maximum integer such that all agents are in epoch $i$ in configuration $C_{t_e}$,
and let $V_k$ be the set of agents whose rank is in $S_k$ in configuration $C_0$.
Since $|V_k| \leq 2\ell$, during epoch $i$,
\ie, in configurations $C_{t_b}, C_{t_b+1}, \dots, C_{t_e}$,
there are always at most $2 \ell \cdot 2^{\childmax} \leq n / 2$ non-null agents,
thus there are at least $n / 2$ null agents.

We will prove the lemma using the concept of \emph{virtual agents} introduced by \cite{SNY+12} with slight modifications.
Let $U(t)$ be the set of null agents at time step $t$. From the above discussion,
$U(t) \geq n/2$ holds for all $t \in [t_b, t_e]$.
Define $\iota(v) = C_{t_b}(v).\children$. 
If proliferation completes correctly in epoch $i$,
each $v \in V\setminus U(t_b)$ must have generated $2^{\iota(v)}$ non-null agents, including itself.
In this proof, we map those agents to distinct integers from $0$ to $2^{\iota(v)}-1$.
For any $v \in V \setminus U(t_b)$, $\alpha \ge 0$, and $j \ge 1$, we denote 
the $j$-th bit from the end of $\alpha$ by $b(v,\alpha,j)$, \ie, $b(v,\alpha,j)=\lfloor \alpha / 2^{j-1} \rfloor \bmod 2$.
By definition, $b(v,\alpha,j) = 0$ for any $j \ge \iota(v)+1$.
Then, for any $t \geq t_b$, we define the \emph{virtual agent} $\nu(v,\alpha,t)$ and its number of divisions $\kappa(v,\alpha,t)$ as follows:
$$ \nu(v,\alpha,t_b)=v, \quad \kappa(v,\alpha,t_b)=0,$$
and for any $t \ge t_b$, 
\begin{align*}
\nu(v,\alpha,t+1)
&=
\begin{cases}
u & \text{if }  b(v,\alpha,\kappa(v,\alpha,t)+1)=1 \land \exists u \in U(t): \Gamma_{t} = (\nu(v,\alpha,t),u)\\
\nu(v,\alpha,t) & \text{otherwise},
\end{cases}
\\
\kappa(v,\alpha,t+1)
&=
\begin{cases}
\min(\kappa(v,\alpha,t)+1,\iota(v)) & \text{if } \exists u \in U(t): \Gamma_{t} =(\nu(v,\alpha,t),u)\\
\kappa(v,\alpha,t) & \text{otherwise}.
\end{cases}
\end{align*}

\noindent 
Clearly, proliferation completes before time step $t$ if $\kappa(v,\alpha,t) = \iota(v)$ for all $v \in V \setminus U(t_b)$ and $\alpha \in [0, 2^{\iota(v)}-1]$. Note that a virtual agent $\nu(v,\alpha,t)$ is determined solely by $C_{t_b}$ and $\Gamma_{t_b}, \Gamma_{t_b+1}, \dots, \Gamma_{t-1}$. Therefore, at any time step $t \in [t_b, t_e]$, the virtual agent $\nu(v,\alpha,t)$ interacts with a non-null agent as the initiator with a probability of $\frac{1}{n} \cdot \frac{U(t)}{n-1} \geq \frac{1}{2(n-1)}$. Consequently, by the Chernoff Bound, $\kappa(v,\alpha,t) = \iota(v)$ holds with high probability for some $\tdone = t_b + O(n \log n)$. By applying the union bound, with high probability, this condition is met for all $v \in V \setminus U(t_b)$ and $\alpha \in [0, 2^{\iota(v)}-1]$, indicating that the proliferation completes within $O(n \log n)$ steps with high probability. Note that we can disregard the probability that $t_e < \tdone$, as we can choose arbitrarily large constant $d_1$ in Lemma \ref{lem:phase_clock}.
 \end{proof}

\begin{lemma}
\label{lem:coldb}
Let $P=\coldb(\nlower,\nupper)$.
Starting from any initialized configuration $C_0 \in \cinit(P)$,
execution $\Xi_{P}(C_0,\rs)$ solves the collision detection.
It stabilizes within 
$O(n^{3/2} \cdot \sqrt{\log n})$ time steps (respectively, $O(n^{3/2} \cdot \log^{3/2} n)$ time steps) in expectation
(respectively, with high probability).
This protocol uses $\tilde{O}(n)$ states per agent.
\end{lemma}

\begin{proof}
We can easily verify that the number of states is $\tilde{O}(n)$ by checking the domain of the variables, excluding the input variable $\rank$: $\grid$ uses $O(\sqrt{n})$ states and $\epoch$ uses $O(r\cdot z)=\tilde{O}(\sqrt{n})$ states, while all other variables use only a poly-logarithmic number of states.

Proving the correctness is also easy. The first part of the protocol (\ie, the backup protocol)
eventually makes an agent raise a flag if and only if there is at least one colliding pair of agents.
The raised flag is propagated to the entire population via the epidemic protocol. 
False-positive errors are not permitted also in the second part of $\coldb$.
In each epoch $i \in [1, r \cdot z]$, the group identifiers are initialized. When no agents share a common rank, all non-null agents $v$ with $\pi_1(v.\grid) = \alpha$ originate from a common agent with a rank $k \cdot \ell + \alpha$, where $k = (i-1) \bmod z$. Therefore, all such agents have the same nonce, thus preventing any of them from raising a collision flag.
Thus, starting from any initialized configuration $C_0 \in \cinit(P)$,
execution $\Xi_{P}(C_0,\rs)$ eventually reaches a stable configuration with probability $1$,
where all agents output the correct answer.

Next, we bound the stabilization time of $P$. From the above discussion, the stabilization time is zero when there is no rank collision in $C_0$. 
Thus, assume the presence of a colliding pair $u, v \in V$, \ie, $u.\rank = v.\rank$, in $C_0$.
Let $k \in \{0,1,\dots,z-1\}$ be the integer such that $C_0(u).\rank \in S_k$.
By Lemma \ref{lem:not_so_many_colliding_pairs}, we assume that there are only $o(\sqrt{n\log n})$ colliding pairs of agents without loss of generality. 
Under this assumption, $|V_k| \le 2\ell$,
where $V_k$ is the set of agents whose ranks are in $S_k$ in configuration $C_0$.
Otherwise, there are $\sum_{j=1}^{\ell} \binom{x_j}{2} \ge \sum_{j=1}^{\ell} (x_j-1)>\ell = \Omega(\sqrt{n\log n})$ colliding pairs, yielding a contradiction,
where $x_j$ is the number of agents whose rank is $k\cdot \ell+j$. 
With probability $1/2$, $u$ and $v$ generate different nonces and thus obtain distinct group identifiers when they enter epoch $i$.
By Lemma \ref{lem:proliferation}, with high probability,
each of their group identifiers is copied to $2^{\childmax}\ge \nlower/(8\ell)$ 
agents, respectively, in the first half period of epoch $i$, at least $\frac{d_1}{2}  n \log n$ time steps, where 
$d_1$ is the constant that appears in Lemma \ref{lem:phase_clock}.
Thereafter, during the next $\frac{d_1}{2} \cdot n \log n$ time steps, there is at least one time step at which
one agent in $u$'s group and one agent in $v$'s group have an interaction
with probability at least
$$1-\left(1- \frac{\binom{\nlower/8\ell}{2}}{\binom{n}{2}}\right)^{\frac{d_1}{2} \cdot n \log n}
= 1 - \left (1-\Omega \left ( \frac{1}{n \log n}\right ) \right)^{\frac{d_1}{2} \cdot n \log n} \ge 1-\epsilon$$
for any constant $\epsilon$ because we can choose an arbitrarily large constant for $d_1$.
To conclude, some agent raises a collision flag in epoch $i$ with probability $1/2 -\epsilon \ge 1/3$.
There are $r$ such epochs dedicated to $S_k$, and those epochs appear in every $z$ epochs.
Thus, $\Xi(C_0,\rs)$ stabilizes within $O(r\cdot z \cdot n \log n)= O(n^{3/2} \log^{3/2} n)$ steps with probability 
at least $1-(1-1/3)^{r} \ge 1-2^{-r} \ge 1-n^{-\eta}$, \ie, with high probability.
The expected stabilization time is bounded by
$$\sum_{j=1}^{r} \left (1-\frac{1}{3}\right )^{j-1} \cdot \frac{j  z}{3} \cdot O(n \log n)
+ \left(1-\frac{1}{3} \right ) ^r \cdot O(n^2) = O(z n \log n)=O\left (n^{\frac{3}{2}} \sqrt{\log n} \right ).$$
 \end{proof}

\begin{remark}[No false-positive error even without the assumptions]
\label{remark:coldb}
The epoch $v.\epoch$ of each agent $v \in V$ is monotonically non-decreasing.
Therefore, no false-positive errors occur in the execution of $\coldb(\nlower, \nupper)$ even without the previously stated assumption that there is always a unique leader $a_L$ and the parameters satisfy $\nlower \leq n \leq \nupper$, $\nlower = \Theta(n)$, and $\nupper = \Theta(n)$.
\end{remark}

\begin{remark}[Derandomization]
\label{remark:derandomization}
One might consider $\coldb$ a randomized algorithm because it generates a binary nonce at Line 16, seemingly deviating from the model defined in Section \ref{sec:model}. However, we can easily derandomize our protocol by exploiting the randomness from the uniformly random scheduler $\rs$. The derandomization can proceed as follows:
(i) When an agent increases its epoch, it enters the \emph{waiting mode} instead of initializing $\grid$ and $\children$;
(ii) When an agent in the waiting mode has an interaction,
the agent can generate a 1-bit nonce $\chi$, setting $\chi = 0$ if it is the initiator of the interaction, and $\chi = 1$ otherwise.,
thus it initializes $\grid$ and $\children$ with the generated nonce,
and returns to the \emph{normal mode}.
Let $u$ and $v$ be any two agents with the same rank. Note that the nonce generated by $u$ and $v$ may not be independent, which could potentially complicate the analysis. However, their nonces are dependent only if they have previously met, at which point they would have raised a collision flag via the backup protocol. Therefore, this dependency does not pose a problem.
\end{remark}

\begin{algorithm}[t]
\caption{
$\cold$ at an interaction where
$\aaa$ and $\bbb$ are an initiator and a responder, respectively.
} 
\label{al:cold}
 Execute the counting protocol \approximate \cite{BKR19}\;
\If(\tcp*[f]{$\aaa$ never updates its $\level$}){$\bbb.\level$ increased at this interaction}{
Reset all variables of $\bbb$ maintained by $\coldb$ and its submodules
to their initial values.
}
\If{$
(\aaa.\level, \aaa.\lognum)=(\bbb.\level, \bbb.\lognum)
\land
\aaa.\countfin = \bbb.\countfin = 1
$
}
{
Execute $\coldb(r,2^{\aaa.\lognum-1},2^{\aaa.\lognum+1})$
}
 \If{$\aaa.\rank = \bbb.\rank$}{$\bbb.\collision \gets 1$\;}
 Execute $\epidemic(\collision)$\;
\end{algorithm}

\subsection{Collision Detection without Assumptions}
\label{sec:with}

To run $\coldb(\nlower,\nupper)$, we require a unique leader $a_L \in V$ and two parameters, $\nlower$ and $\nupper$, that satisfy $\nlower \leq n \leq \nupper$, $\nlower = \Theta(n)$, and $\nupper = \Theta(n)$. We address the collision detection problem without this assumption
by running this protocol and the counting protocol $\approximate$, given by Berenbrink, Kaaser, and Radzik~\cite{BKR19}, in parallel. The protocol $\approximate$ elects a leader using the leader election protocol presented by G{\k{a}}sieniec and Stachowiak~\cite{GS18}, and computes $\lfloor \log n \rfloor$ or $\lceil \log n \rceil$ with high probability.
This protocol uses $O(\log n \cdot \log \log n)$ states with high probability.

The protocol $\approximate$ maintains four variables: $\lognum \in \mathbb{N}$, $\level \in \mathbb{N}$, $\countfin \in \{0,1\}$, and $\leader \in \{0,1\}$~\footnote{
We have changed the names of the variables for consistency with the notations of this paper.
}.
The following proposition holds by the definition of $\approximate$, which enables us to integrate these two protocols. (See \cite{BKR19} for the definition of $\approximate$.)
\begin{proposition} 
\label{prop:approximate}
An execution of the protocol $\approximate$ has the following properties:
\begin{itemize}
    \item For all $v \in V$, $v.\level$ is monotonically non-decreasing 
    and increases only when $v$ joins an interaction as the responder,    
    \item For all $v \in V$, $v.\lognum$ never changes while $v.\countfin = 1$, 
    \item For all $v \in V$,
    $v.\countfin$ reverts from $1$ to $0$ only if $v.\level$ increases, and
    \item When two agents with the same level meet, the $\countfin$ value is updated according to the one-way epidemic from the initiator to the responder. Only an agent with $\leader = 1$ can set $\countfin$ to $1$, in addition to this rule,    
\end{itemize}
\end{proposition}

For any variable $\var$, we say that a configuration $C$ is $\var$-stable if, for every configuration $C'$ reachable from $C$ and every agent $v \in V$, the value of $\var$ for $v$ is the same in both $C$ and $C'$, \ie, $C(v).\var = C'(v).\var$.
For any set of variables $\chi$, we say that a configuration $C$ is $\chi$-stable if $C$ is $\var$-stable for every $\var$ in $\chi$.
Due to the analysis in \cite{BKR19}, we have the following lemma:
\begin{lemma}[\cite{BKR19}]
\label{lem:approximate}
Let $C_0 \in \cinit(P)$ be any initialized configuration and let $\Xi_{P}(C_0, \rs) = C_0, C_1, \dots$,
where $P=\approximate$.
With high probability, there is an integer $t = O(n\log^2 n)$ such that:
\begin{itemize}
\item $C_{t}$ is a $\{\lognum, \leader, \level\}$-stable configuration,
\item All agents share the same value in $\lognum$, which is either $\lfloor \log n \rfloor$ or $\lceil \log n \rceil$, in $C_t$,
\item There exists exactly one agent $a_L$ with $\leader = 1$ in $C_t$,
\item All agents satisfy $\countfin = 0$ in $C_{t-1}$,
\item Exactly two agents satisfy $\countfin = 1$, one of whom is $a_L$, in $C_{t}$.
\end{itemize}
The protocol $\approximate$ uses $O(\log n \cdot \log \log n)$ states with high probability. 
\end{lemma}


If agents could determine whether they have already reached a configuration $C_t$, integrating $\coldb$ and $\approximate$ would be straightforward: when two agents $\aaa$ and $\bbb$ meet, they execute $\approximate$ if either $\aaa.\countfin = 0$ or $\bbb.\countfin = 0$; otherwise, they execute $\coldb(\nlower,\nupper)$, setting $\nlower = 2^{\aaa.\lognum - 1}$ and $\nupper = 2^{\aaa.\lognum + 1}$. However, such straightforward integration is not feasible, as agents cannot detect this condition. In particular, the value of the variable $\countfin$ in $\approximate$ may revert from $1$ back to $0$.


The pseudocode of $\cold$ is presented in Algorithm \ref{al:cold}, which utilizes the properties outlined in Proposition \ref{prop:approximate} and \ref{lem:approximate} to integrate $\coldb$ and $\approximate$.
Suppose that two agents $\aaa$ and $\bbb$ have an interaction where they serve as the initiator and the responder, respectively. Initially, they execute $\approximate$ (Line 29). As specified in Proposition \ref{prop:approximate}, $\aaa.\level$ does not increase during $\approximate$. If $\bbb.\level$ does increase, all variables maintained by $\coldb$ for $\bbb$ are reset to their initial values (Lines 30, 31). In essence, each time an agent's level increases, it terminates the current execution of $\coldb$ and restarts it.
Subsequently, $\aaa$ and $\bbb$ execute $\coldb(\nlower, \nupper)$ with $\nlower = 2^{\aaa.\lognum - 1}$ and $\nupper = 2^{\aaa.\lognum + 1}$, but only if they have matching values in both $\level$ and $\lognum$, and both observe $\countfin = 1$. According to Lemma \ref{lem:approximate}, these conditions will eventually be met with high probability.
However, there remains a small but non-zero probability that $\approximate$ fails to meet these conditions, which would prevent $\coldb$ from executing. Thus, we also implement a backup protocol outside of $\coldb$ (Lines 34–36), eliminating the risk of false-negative errors.


\begin{theorem}
\label{theorem:main}
Protocol $\cold$ solves the collision detection problem with probability $1$. 
The execution of $\cold$ under the uniformly random scheduler
stabilizes within $O(n^{3/2} \log^{3/2} n)$ time steps and uses $\tilde{O}(n)$ states per agent (excluding the space needed to store the input variable $\rank$) with high probability. It stabilizes within $O(n^{3/2} \sqrt{\log n})$ time steps in expectation.
\end{theorem}
\begin{proof}
The protocol $\cold$ is always correct. Specifically, the execution of $\cold$ eventually stabilizes to correctly output:
$1$ if any agents share the same rank, and $0$ otherwise.
This correctness immediately follows from the fact that (i) a backup protocol (Lines 34--36) eventually computes the correct answer, and (ii) $\coldb$ (Line 33) does not cause false-positive errors (see Remark \ref{remark:coldb}).

Next, we analyze the stabilization time. Let $C_0 \in \cinit(P)$ be any initialized configuration, and let $\Xi = \Xi_{P}(C_0, \rs) = C_0, C_1, \dots$, where $P = \cold$. According to Lemma \ref{lem:approximate}, with high probability, there exists $t = O(n \log^2 n)$ that meets the conditions specified by Lemma \ref{lem:approximate}. Since $C_t$ is $\{\lognum, \leader, \level\}$-stable, in the suffix $C_t, C_{t+1}, C_{t+2}, \dots$ of the execution, no agent changes the values of $\lognum$, $\leader$, or $\level$. Therefore, the assumptions made for $\coldb$ are always satisfied in this suffix. At interaction $\Gamma_{t-1}$, which transitions $C_{t-1}$ to $C_t$, the unique leader $a_L$ and one non-leader agent, say $u \in V$, raise the $\countfin$ flag and execute $\coldb$ for the first time since their level reaches $C_t(a_L).\level$. Thereafter, as stated in Proposition \ref{prop:approximate}, other agents also raise the $\countfin$ flag according to the one-way epidemic and never revert the flag from $1$ to $0$.
Moreover, in $C_t$, all agents except $a_L$ and $u$ retain the initial values for all variables maintained by $\coldb$, while during interaction $\Gamma_{t-1}$, $a_L$ and $u$ simulate one interaction of $\coldb$ with those initialized variables. Therefore, the suffix $C_t, C_{t+1}, C_{t+2}, \dots$ exactly corresponds to the execution of $\coldb$ because in an execution of $\coldb$, no agent changes its state before the unique leader $a_L$ has its first interaction. Specifically, an agent $v$ updates its states by interaction $\Gamma_{t'}$ only if there is an increasing sequence of integers $t_1, t_2, \dots, t_{k-1}$ and a sequence of agents $v_1, v_2, \dots, v_{k}$ such that: 
(i) $v_1 = v_L$, (ii) $v_k = v$, (iii) $t_{k-1} \le t'$,
and (iv) $\Gamma_{t_i} = (v_i, v_{i+1})$ for each $i \in [1, k-1]$.
Therefore, by Lemmas \ref{lem:coldb} and \ref{lem:approximate}, $\cold$ stabilizes within $O(n^{3/2} \log^{3/2} n)$ time steps with high probability and $O(n^{3/2} \sqrt{\log n})$ time steps in expectation.

The claim for the number of states immediately follows from Lemmas \ref{lem:coldb} and \ref{lem:approximate}.
 \end{proof}

\paragraph*{Acknowledgments} 
This work is supported by JST FOREST Program JPMJFR226U and 
JSPS KAKENHI Grant Numbers JP19K11826, JP20H04140, and JP20KK0232.

\bibliography{population}



\end{document}